\documentclass[a4paper,onecolumn,11pt,aps,nofootinbib,superscriptaddress,tightenlines]{scrartcl}

\usepackage{amsmath}
\usepackage{amssymb}  
\usepackage{amsthm}
\usepackage{amsfonts}
\usepackage{graphicx}
\usepackage{bbm}
\usepackage{url}
\usepackage{ upgreek }
\usepackage{color}

\usepackage{authblk}

\newcommand{\tr}{\textnormal{tr}}

\newcommand{\ket}[1]{| #1 \rangle}

\newcommand{\proj}[2]{| #1 \rangle\!\langle #2 |}

\def\beq{\begin{equation}}
\def\eeq{\end{equation}}
\def\bq{\begin{quote}}
\def\eq{\end{quote}}
\def\ben{\begin{enumerate}}
\def\een{\end{enumerate}}
\def\bit{\begin{itemize}}
\def\eit{\end{itemize}}

\def\ra{\rightarrow}

\def\lb{\left(}
\def\rb{\right)}
\def\lset{\lbrace}
\def\rset{\rbrace}

\def\l|{\left|}
\def\r|{\right|}
\def\lbr{\left[}
\def\rbr{\right]}
\def\ident{\textnormal{id}}

\newcommand\C{\mathbbm{C}}

\newcommand\R{\mathbbm{R}}
\newcommand\N{\mathbbm{N}}
\newcommand\M{\mathcal{M}}

\newcommand{\U}{\mathcal{U}}

\newcommand{\ketbra}[1]{|#1\rangle\langle#1|}

\theoremstyle{plain}
\newtheorem{thm}{Theorem}[section]
\newtheorem{lem}[thm]{Lemma}
\newtheorem{cor}[thm]{Corollary}
\theoremstyle{definition}

\begin{document}
\title{\textbf{All unital qubit channels are $4$-noisy operations}}

\author{Alexander M\"uller-Hermes}
\author{Christopher Perry}

\affil{\small{QMATH, Department of Mathematical Sciences, University of Copenhagen, Universitetsparken 5, 2100 Copenhagen, Denmark}}

\maketitle
\date{\today}

\begin{abstract}
We show that any unital qubit channel can be implemented by letting the input system interact unitarily with a $4$-dimensional environment in the maximally mixed state and then tracing out the environment. We also provide an example where the dimension of such an environment has to be at least $3$.  
\end{abstract}

Within the set of complex $d\times d$ matrices $\M(\C^d)$ we denote the set of unitary matrices by $\mathcal{U}\lb\C^d\rb$. Physically, unitary matrices correspond to time evolutions of a closed quantum system via the linear map $\rho\mapsto U\rho U^\dagger$ acting on quantum states, i.e.~positive matrices $\rho\in (\M(\C^d))^+$ with $\text{Tr}(\rho)=1$. When a quantum system interacts with an environment, general time evolutions are given by quantum channels $T:\M(\C^d)\ra \M(\C^d)$, i.e.~completely positive and trace preserving linear maps. In the following we will only consider unital quantum channels, i.e.~where $T(\mathbbm{1}_d) = \mathbbm{1}_d$ for the unit matrix $\mathbbm{1}_d\in \M(\C^d)$. 

We will consider two natural subclasses of unital quantum channels: Mixed unitary channels and $n$-noisy operations. A quantum channel $T:\M(\C^d)\ra \M(\C^d)$ is called a mixed unitary channel if there exists a $k\in \N$, unitaries $U_i\in \U\lb\C^d\rb$ for $i\in \lset 1,\ldots ,k\rset$ and a probability distribution $\lset p_i\rset^k_{i=1}\subset \R_{\geq 0}$ such that 
\[
T(X) = \sum^k_{i=1} p_i U_iXU^\dagger_i
\] 
for any $X\in \M(\C^d)$. Physically, a mixed unitary quantum channel corresponds to the random application of a unitary channel $\rho\mapsto U_i\rho U_i^\dagger$ with probability $p_i$. 

A quantum channel $T:\M(\C^d)\ra \M(\C^d)$ is called an $n$-noisy operation if there exists a bipartite unitary $U\in \U\lb\C^d\otimes \C^n\rb$ such that 
\begin{equation}
T(X) = \text{Tr}_E\lb U\lb X\otimes \frac{\mathbbm{1}_n}{n}\rb U^\dagger\rb
\label{equ:nNoisy}
\end{equation}
for any $X\in \M(\C^d)$. Here $\text{Tr}_E = \ident_{d}\otimes \text{Tr}$ denotes the partial trace on the second tensor factor (i.e.~the system corresponding to the environment). Physically, an $n$-noisy operation arises from letting the system interact unitarily with an $n$-dimensional environment in the maximally mixed state and then forgetting about the environment system. 

To our knowledge, the class of $n$-noisy operations has been first introduced in the context of quantum thermodynamics~\cite{horodecki2003reversible}. Since the maximally mixed state is the thermal state of a Hamiltonian proportional to $\mathbbm{1}_n$, this class can be considered as a toy model for the possible time evolutions of a system interacting with a simple thermal bath. Within the quantum information theory literature, such maps also appear in the study of open quantum systems under the guise of $k$-unistochastic maps \cite{zyczkowski2004duality,musz2013unitary} which are identical to $d^k$-noisy operations. 

An important generalization of $n$-noisy operations is given by the class of factorisable maps. These have been introduced in \cite{anantharaman2006ergodic} in the general framework of operator algebras and later studied in the special case of matrix algebras, eventually leading to a counterexample to the asymptotic quantum Birkhoff conjecture \cite{haagerup2011factorization,haagerup2015asymptotic}. Such channels have a form similar to \eqref{equ:nNoisy} with the ``environment'' modeled by a finite von Neumann algebra $\mathcal{N}$ with its unit replacing $\mathbbm{1}_d$ in \eqref{equ:nNoisy} and with a normal faithful tracial state on $\mathcal{N}$ replacing the normalized trace (we refer to \cite{haagerup2011factorization,haagerup2015asymptotic} for more details on this class of maps). It has been shown in \cite{haagerup2011factorization} that any mixed unitary channel is a factorisable map (factorizing in an appropriate sense through an abelian von Neumann algebra). For dimensions $d\geq 3$, the factorisable maps form a strict subset of the unital quantum channels and the Holevo-Werner channel $T:\M(\C^3)\ra \M(\C^3)$ given by
\[
T(X) = \frac{1}{2}\lb \text{Tr}(X)\mathbbm{1}_3 - X^T\rb
\]
is an example of a unital channel that is not factorisable~\cite{landau1993birkhoff}. 

The above results leave open the question regarding the relationship between the sets of mixed unitary quantum channels and $n$-noisy operations, i.e.~whether any mixed unitary quantum channel can be written as in \eqref{equ:nNoisy}. Note that in the context of state transformations, $n$-noisy operations are as powerful as unital channels: For any fixed quantum states $\rho,\sigma\in \M(\C^d)^+$, there exists a unital quantum channel mapping $\rho$ to $\sigma$ iff there exists a $d$-noisy operation achieving the same \cite{chefles2002quantum, horodecki2003reversible, scharlau2016quantum}. The problem is much harder when equality of the channels is required. 

Recently, this problem was solved by Musat \cite{Magdalena2017} with examples of mixed unitary channels in all dimensions $d\geq 3$ that are not $n$-noisy operations for any $n\in\N$. It has been open whether such counterexamples could also exist for dimensions $d=2$. Note that for $d=2$ the sets of unital quantum channels and mixed unitary channels coincide \cite{ruskai2002analysis} and in this case the question can be rephrased to: Is every unital qubit quantum channel an $n$-noisy operation? 

In Section \ref{sec:Main} we answer the above question, showing that any unital quantum channel $T:\M(\C^2)\ra \M(\C^2)$ is a $4$-noisy operation, i.e.~can be written as \eqref{equ:nNoisy} with $n=4$. It has been previously noted by Musz et al.~\cite{musz2013unitary} that $n=2$ is not enough to achieve this. For completeness we present this result in a cleaner form in Section \ref{sec:Kraus}. We leave open the question as to whether every unital qubit quantum channel is a $3$-noisy operation.

\section{Main results}
\label{sec:Main}

Our proof follows a general idea of Choi and Wu~\cite{CHOI199025} on representing matrices as convex combinations of rank-$k$ projections. We begin with a lemma inspired by~\cite[Lemma 3.2]{CHOI199025}   

\begin{lem}
Let $a,b,c,d\in\R_{\geq 0}$ be such that $a\geq b\geq c\geq d$ and $a+d = b+c$. For any unitaries $U_1,U_2\in\U(\C^2)$ there exist unitaries $V_1,V_2\in \U(\C^2)$ such that 
\[
a U_1XU^\dagger_1 + dU_2XU^\dagger_2 = bV_1XV^\dagger_1 + cV_2XV^\dagger_2
\]
for any $X\in \M(\C^2)$.
\label{lem:Trick}
\end{lem}

\begin{proof}
We can assume without loss of generality that $d>0$ since otherwise the statement is trivial using $V_1=V_2=U_1$. By the spectral theorem, we can find a diagonal unitary $D\in \U(\C^2)$ such that $D = S^\dagger U_1^\dagger U_2S$ for a unitary matrix $S\in\U(\C^2)$. We will construct $W_1,W_2\in \U(\C^2)$ such that 
\begin{equation}
aX + dDXD^\dagger = bW_1XW^\dagger_1 + cW_2XW^\dagger_2
\label{equ:Bla}
\end{equation}   
for any $X\in\M(\C^2)$. Given such unitaries, setting $V_1 = U_1SW_1S^\dagger$ and $V_2 = U_1 SW_2S^\dagger$ would finish the proof. 

For $D = \text{diag}(z_1, z_2)$, we set $z_1\bar{z_2} = \text{exp}(i\theta)$ with $\theta\in\left[0,2\pi\right)$. Now we make the ansatz $W_1 = \text{diag}(\text{exp}(i\alpha),1)$ and $W_2 = \text{diag}(\text{exp}(i\beta),1)$ with $\alpha,\beta\in\left[ 0,2\pi\rb$. Inserting these matrices and using $a+b=c+d$, we find that \eqref{equ:Bla} holds iff
\[
a + d~\text{exp}(i\theta) - b~\text{exp}(i\alpha) = c~\text{exp}(i\beta).
\] 
Given $\theta\in\left[0,2\pi\right)$, we can find a $\alpha,\beta\in\left[ 0,2\pi\rb$ satisfying the previous equation if we can find $\alpha\in\left[ 0,2\pi\rb$ with 
\[
\left| a + d~\text{exp}(i\theta) - b~\text{exp}(i\alpha)\right| = c.
\] 
The previous equation is equivalent to 
\begin{equation}
a^2 + b^2 - c^2 + d^2 + 2ad~\text{cos}(\theta) - 2ab~\text{cos}(\alpha) - 2db\text{cos}(\theta -\alpha) = 0.
\label{equ:bla2}
\end{equation}
To show the existence of $\alpha\in\left[ 0,2\pi\rb$ satisfying \eqref{equ:bla2}, we note that for $\alpha=0$ the left hand side is
\[
(a-b)^2 - c^2 + d^2 + 2d(a-b)~\text{cos}(\theta) \leq (a-b)^2 - c^2 + d^2 + 2d(a-b) = 0. 
\] 
Similar for $\alpha = \theta$ the left hand side is
\[
(b-d)^2 + a^2 - c^2 - 2a(b-d)~\text{cos}(\theta) \geq (b-d)^2 + a^2 - c^2 - 2a(b-d) = 0. 
\] 
By continuity, there exists an $\alpha\in\left[ 0,\theta\right]$ such that \eqref{equ:bla2} is satisfied. This finishes the proof. 

\end{proof}

Before we state our main theorem we need some additional notation. Given a probability distribution $p=\lset p_i\rset^k_{i=1}\subset \R_{\geq 0}$, let $p^\downarrow\in (\R_{\geq 0})^k$ denote the vector obtained from $p$ by reordering the probabilities $p_i$ in decreasing order, i.e. such that $1\geq p^\downarrow_1\geq \cdots \geq p^\downarrow_k \geq 0$. We say that $p$ \emph{majorizes} another probability distribution $q=\lset q_i\rset^l_{i=1}\subset \R_{\geq 0}$ iff 
\[
\sum^j_{i=1}p^{\downarrow}_i\geq \sum^j_{i=1}q^{\downarrow}_{i}
\]  
for any $j\in \lset 1,\ldots , \min(k,l)\rset$, and we will write $p\succ q$ to denote this relation.\footnote{Note that $p\succ q$ is equivalent to $\left(q_1,\dots,q_k\right)\in\textrm{Conv}\left\{\left(p_{\sigma(1)},\dots,p_{\sigma(k)}\right):\sigma\in S_k\right\}$ where we append zeros to $p$ or $q$ as necessary to make them the same size (see \cite[Corollary B.3.]{marshall1979inequalities}).} The proof of the following theorem follows the lines of a standard argument in the theory of majorization (see \cite[Lemma B.1.]{marshall1979inequalities}), relating majorization to simple transformations known as T-transforms. We will repeat it here for the sake of completeness.

\begin{thm}
Suppose $T:\M(\C^2) \ra \M(\C^2)$ is given by 
\[
T(X) = \sum^k_{i=1} p_i U_iXU^\dagger_i
\]
for a probability distribution $p=\lset p_i\rset^k_{i=1}\subset \R_{\geq 0}$ and unitaries $U_1,U_2,\ldots ,U_k\in \U(\C^2)$. Then, for any probability distribution $q=\lset q_i\rset^l_{i=1}\subset \R_{\geq 0}$ satisfying $p\succ q$, there exist unitaries $V_1,V_2,\ldots ,V_l\in \U(\C^2)$ such that 
\[
T(X) = \sum^l_{i=1} q_iV_iXV^\dagger_i
\]
for any $X\in\M(\C^2)$.
\label{thm:Averaging}
\end{thm}
\begin{proof}
By appending zero probabilities to either $p$ or $q$ we can assume without loss of generality that $k=l$. We can also assume that $1\geq p_1\geq \cdots \geq p_k \geq 0$ and that $1\geq q_1\geq \cdots \geq q_k \geq 0$, and that $p\neq q$ (otherwise the statement is trivial). Since $p\succ q$, we know that the largest index $i$ for which $q_i\neq p_i$ satisfies $q_i>p_i$ and the we can define indices
\[
n:=\max\lset i~:~q_i < p_i\rset\hspace*{0.3cm} \text{ and } \hspace*{0.3cm}m:=\min\lset i>n~:~q_i > p_i\rset.
\]
Setting $\delta := \min(p_n-q_n,q_m - p_m)$, it is easy to check that
\[
p_n>p_n-\delta\geq q_n\geq q_m \geq p_m+\delta >p_m.
\]
Now we can apply Lemma \ref{lem:Trick} for
\[
a=p_n,\hspace{0.5cm} b=p_n-\delta,\hspace{0.5cm} c=p_m+\delta,\hspace{0.5cm} d=p_m,
\]
to find unitaries $V, W\in \U(\C^2)$ such that 
\[
p_n U_nX U^\dagger_n + p_m U_mXU^\dagger_m = (p_n-\delta)VXV^\dagger + (p_m+\delta)WXW^\dagger
\]
for any $X\in \M(\C^2)$. Thus, replacing the unitaries $U_n,U_m$ by $V,W$ and changing the coefficients accordingly yields a new expansion of $T$. Note that by definition of $n$ and $m$, for any $i\in \lset 1,\ldots, k\rset$ satisfying $n<i<m$ we have $q_i=p_i$ and the ordering
\begin{equation}
p_n>p_n-\delta\geq q_n\geq q_i = p_i \geq q_m \geq p_m+\delta >p_m
\label{equ:ordering}
\end{equation}
holds. This shows that the new distribution 
\[
(p_1,\ldots , p_{n-1}, p_n -\delta , p_{n+1} , \ldots , p_{m-1},p_m + \delta ,p_{m+1},\ldots ,p_k)
\]
is still decreasingly ordered and by \eqref{equ:ordering} this distribution still majorizes $q$. Note furthermore that either $p_n-\delta = q_n$ or $p_m+\delta =q_m$, so the the number of indices where the new distribution coincides with the `target' distribution $q$ has increased at least by $1$. Repeating this process at most $k-2$ more times finishes the proof.

\end{proof}

Since any probability distribution $p=\lset p_i\rset^k_{i=1}$ satisfies $p\succ \left(\frac{1}{k},\frac{1}{k},\ldots,\frac{1}{k}\right)$, we have the following corollary.


\begin{cor}
Suppose $T:\M(\C^2) \ra \M(\C^2)$ is given by 
\[
T(X) = \sum^k_{i=1} p_i U_iXU^\dagger_i
\]
for a probability distribution $\lset p_i\rset^k_{i=1}\subset \R_{\geq 0}$ and unitaries $U_1,U_2,\ldots ,U_k\in \U(\C^2)$. Then there exist unitaries $V_1,V_2,\ldots ,V_k\in \U(\C^2)$ such that 
\[
T(X) = \frac{1}{k}\sum^k_{i=1} V_iXV^\dagger_i
\]
for any $X\in\M(\C^2)$.
\label{cor:mean}
\end{cor}

To state our main result, recall that given a quantum channel $T:\M(\C^{d})\ra \M(\C^{d})$, its Kraus rank is defined as the unique number $R(T)\in\lset 1,\ldots , d^2\rset$ such that
\[
T(X) = \sum^{R(T)}_{i=1} A_iXA^\dagger_i
\]  
for any $X\in \M(\C^{d})$ where $A_i\in \M(\C^d)$ are non-zero, mutually orthogonal\footnote{with respect to the Hilbert-Schmidt inner product.} matrices. Then we obtain:

\begin{cor}
For any unital quantum channel $T:\M(\C^2)\ra \M(\C^2)$ and any $k\geq R(T)$, there exists a unitary $U\in\mathcal{U}\lb \C^2\otimes \C^k\rb$ such that 
\[
T(X) = \text{Tr}_E\lb U\lb X\otimes \frac{\mathbbm{1}_k}{k}\rb U^\dagger\rb
\]
for any $X\in \M(\C^2)$. In particular, since $R(T)\leq 4$, every such quantum channel is a $4$-noisy operation.
\label{cor:MAIN}
\end{cor}

\begin{proof}
By \cite[Section 3.3]{ruskai2002analysis} every unital qubit channel $T:\M(\C^2)\ra \M(\C^2)$ admits a mixed unitary Kraus decomposition, i.e. for any $k\geq R(T)$ there are unitaries $U_1,\ldots , U_k\in\mathcal{U}(\C^{2})$ and a probability distribution $\lset p_i\rset^k_{i=1}\subset \R_{>0}$ such that 
\[
T(X) = \sum^k_{i=1} p_i U_i XU^\dagger_i
\]
for any $X\in\M(\C^2)$. Note that in order to realize $k>R(T)$ we can simply split a non-zero probability, leading to some of the $U_i$ to coincide. Applying Theorem \ref{thm:Averaging} we can find unitaries $V_1,\ldots ,V_k\in\mathcal{U}(\C^{2})$ such that 
\[
T(X) = \frac{1}{k}\sum^k_{i=1} V_i X V^\dagger_i
\] 
for any $X\in\M(\C^2)$. Therefore we have  
\[
T(X) = \text{Tr}_E\lbr U\left(X\otimes \frac{\mathbbm{1}_k}{k}\right) U^\dagger\rbr
\]
with the unitary $U = \sum^k_{i=1} V_i\otimes \ketbra{i}\in \mathcal{U}(\C^2\otimes \C^k)$. 
\end{proof}

Let us briefly compare our Corollary \ref{cor:mean} with the results of~\cite{CHOI199025}. Given a quantum channel $T:\M(\C^d)\ra \M(\C^d)$ consider the eigendecomposition of its (normalized) Choi matrix (see \cite{choi1975completely})
\[
C_T = \sum^{R(T)}_{i=1}p_i \proj{\psi_i}{\psi_i}
\]
for a probability distribution $\lset p_i\rset^{R(T)}_{i=1}\subset \R_{>0}$ and orthonormal vectors $\ket{\psi_i}\in \C^d$. Applying~\cite[Corollary 3.3]{CHOI199025} to this convex combination of rank-$1$ projectors leads to a decomposition of $C_T$ as an average of rank-$1$ projectors (i.e. with equal weights $1/R(T)$). Using the Choi-Jamiolkowski isomorphism to relate the rank-$1$ projectors appearing in this decomposition to completely positive maps of the form $X\mapsto AXA^\dagger$ yields the decomposition
\[
T(X) = \frac{1}{R(T)}\sum^{R(T)}_{i=1} A_iXA^{\dagger}_i 
\]
for any $X\in \M(\C^d)$ with matrices $A_i\in \M(\C^d)$ satisfying $\text{Tr}(A^\dagger_i A_i) = d$. Our proof shows that for $d=2$ and a mixed unitary channel $T:\M(\C^2)\ra \M(\C^2)$ a stronger result holds where the $A_i$ in the above decomposition are unitary. Finally, it should also be noted that our Theorem \ref{thm:Averaging} and Corollary \ref{cor:mean} are similar in flavour to results obtained for convex combinations of unitary matrices~\cite{kadison1986means}.

\section{Kraus ranks of $n$-noisy operations} 
\label{sec:Kraus}

For a bipartite operator $Y\in \M(\C^{d_1}\otimes \C^{d_2})$ the operator Schmidt rank is defined as the unique number $\Omega(Y)\in\N$ such that   
\[
Y = \sum^{\Omega(Y)}_{i=1} A_i\otimes B_i
\]
with orthogonal sets (w.r.t. the Hilbert-Schmidt inner product) of non-zero operators $\lset A_i\rset^{\Omega(Y)}_{i=1}\subset \M(\C^{d_1})$ and $\lset B_i\rset^{\Omega(Y)}_{i=1}\subset \M(\C^{d_2})$. Note that $\Omega(Y)\leq \min(d_1,d_2)^2$.
 
Musz et al.~\cite{musz2013unitary} studied the relation between the operator Schmidt decomposition of a unitary operator $U\in \U(\C^d\otimes \C^n)$ and the Kraus rank of the $n$-noisy operation it generates (cf.~\eqref{equ:nNoisy}). For completeness we will present here the proofs of some results from \cite{musz2013unitary}. For the class of $n$-noisy operations the following holds:

\begin{lem}[\cite{musz2013unitary}]
Let $T:\M(\C^d)\ra \M(\C^d)$ denote an $n$-noisy operation of the form \eqref{equ:nNoisy} with a unitary operator $U\in \U(\C^d\otimes \C^n)$. Then the Kraus rank of $T$ satisfies $R(T)=\Omega(U)$.
\end{lem}

\begin{proof}
By the operator Schmidt decomposition we have
\begin{equation*}
U=\sum_{i=1}^{\Omega(U)} A_i\otimes B_i,
\end{equation*}
where $\lset A_i\rset^{\Omega(Y)}_{i=1}\subset \M(\C^{d_1})$ and $\lset B_i\rset^{\Omega(Y)}_{i=1}\subset \M(\C^{d_2})$ are sets of orthogonal matrices. We can assume without loss of generality that $\tr\left[B_i B^\dagger_j\right]=n\delta_{ij}$. Inserting this decomposition into~\eqref{equ:nNoisy}, we obtain the Kraus decomposition
\begin{equation*}
T\left(X\right)=\sum_{i=1}^{\Omega(U)} A_i X A_i^\dagger.
\end{equation*}
By orthogonality of the $A_i$, we conclude that $R(T)=\Omega(U)$.
\end{proof}

It was shown by D\"ur et al.~\cite{dur2002optimal} that the operator Schmidt rank of a unitary $U\in \U(\C^2\otimes \C^2)$ satisfies $\Omega(U)\in \lset 1,2,4\rset$, i.e. no such unitary can have operator Schmidt rank $3$. It was observed in \cite{musz2013unitary}, that in combination with the previous theorem this implies:

\begin{thm}[\cite{musz2013unitary}]
No $2$-noisy operation $T:\M(\C^2)\ra \M(\C^2)$ has Kraus rank $3$. 
\end{thm}

Note that the unital qubit channel $T:\M(\C^2)\ra \M(\C^2)$ given by 
\begin{equation*}
T\left(X\right)=\frac{1}{3}\sum_{i=1}^{3} \sigma_i X \sigma_i, 
\end{equation*}
has Kraus rank $3$. Here, $\left\{\sigma_i\right\}_{i=1}^{3}$ denote the Pauli matrices\footnote{$\sigma_1=\begin{pmatrix} 0 & 1 \\ 1 & 0
\end{pmatrix}$, $\sigma_2=\begin{pmatrix} 0 & -i \\ i & 0
\end{pmatrix}$, $\sigma_3=\begin{pmatrix} 1 & 0 \\ 0 & -1
\end{pmatrix}$.}. Therefore, by the previous theorem, not every unital qubit channel is a $2$-noisy operation.

Unfortunately, the operator Schmidt decomposition cannot be used to restrict the Kraus rank of $n$-noisy operations $T:\M(\C^d)\ra \M(\C^d)$ for $d>2$ as there are no further obstructions on the operator Schmidt rank of unitary operators~\cite{muller2016restrictions}.

\section*{Acknowledgments}

We thank Magdalena Musat for interesting discussions and comments that helped us to improve this manuscript. We also thank Karol \.Zyczkowski for bringing up a question regarding the relationship between unital qubit channels and $n$-noisy operations at Institut Henri Poincar\'e during the trimester on ``Analysis in Quantum Information Theory''. We acknowledge financial support from the European Research Council (ERC Grant Agreement no 337603), the Danish Council for Independent Research (Sapere Aude) and VILLUM FONDEN via the QMATH Centre of Excellence (Grant No. 10059).

\bibliographystyle{IEEEtran}
\bibliography{mybibliography}

\end{document}